\documentclass{llncs}

\usepackage{graphicx}
\usepackage{latexsym}
\usepackage{epsfig}
\usepackage{fix-cm}
\usepackage{amssymb,amsmath,stmaryrd}
\usepackage{amsfonts}

\def \Sum {{\sf{SUM}}}
\def \pos {{\sf PoS}}
\def \poa {{\sf PoA}}
\def\RP{{\mathbb{R}_{\geq 0}}}
\def\N{{\mathbb{N}}}
\def\R{{\mathbb{R}}}

\begin{document}

\title{On Linear Congestion Games with\\ Altruistic Social Context}
\author{Vittorio Bil\`o}
\institute{ Dipartimento di Matematica e Fisica ``Ennio De Giorgi", Universit\`a del Salento\\ Provinciale Lecce-Arnesano, P.O. Box
193, 73100 Lecce - Italy\\ \textsf{vittorio.bilo@unisalento.it}}

\maketitle

\begin{abstract}
We study the issues of existence and inefficiency of pure Nash equilibria in linear congestion games with altruistic social context, in the spirit of the model recently proposed by de Keijzer {\em et al.} \cite{DSAB13}. In such a framework, given a real matrix $\Gamma=(\gamma_{ij})$ specifying a particular social context, each player $i$ aims at optimizing a linear combination of the payoffs of all the players in the game, where, for each player $j$, the multiplicative coefficient is given by the value $\gamma_{ij}$. We give a broad characterization of the social contexts for which pure Nash equilibria are always guaranteed to exist and provide tight or almost tight bounds on their prices of anarchy and stability. In some of the considered cases, our achievements either improve or extend results previously known in the literature.
\end{abstract}

\section{Introduction}

{\em Congestion games} are, perhaps, the most famous class of non-cooperative games due to their capability to model several interesting competitive scenarios, while maintaining some nice properties. In these games there is a set of players sharing a set of {\em resources}, where each resource has an associated {\em latency function} which depends on the number of players using it (the so-called {\em congestion}). Each player has an available set of strategies, where each strategy is a non-empty subset of resources, and aims at choosing a strategy minimizing her cost which is defined as the sum of the latencies experienced on all the selected resources.

Congestion games have been introduced by Rosenthal~\cite{R73}. He proved that each such a game admits a bounded {\em potential function} whose set of local minima coincides with the set of {\em pure Nash equilibria} of the game, that is, strategy profiles in which no player can decrease her cost by unilaterally changing her strategic choice. This existence result makes congestion games particularly appealing especially in all those applications in which pure Nash equilibria   are elected as the ideal solution concept.

In these contexts, the study of the inefficiency of pure Nash equilibria, usually measured by the sum of the costs experienced by all players, has affirmed as a fervent research direction. To this aim, the notions of {\em price of anarchy} (Koutsoupias and Papadimitriou~\cite{KP99}) and {\em price of stability} (Anshelevich {\em et al.}~\cite{ADKTWR04}) are widely adopted. The price of anarchy (resp. stability) compares the performance of the worst (resp. best) pure Nash equilibrium with that of an optimal cooperative solution.

Congestion games with unrestricted latency functions are general enough to model the Prisoner's Dilemma game, whose unique pure Nash equilibrium is known to perform arbitrarily bad with respect to the solution in which all players cooperate. Hence, in order to deal with significative bounds on the prices of anarchy and stability, some kind of regularity needs to be imposed on the latency functions associated with the resources. To this aim, lot of research attention has been devoted to the case of polynomial latency functions.

In particular, Awerbuch {\em et al.}~\cite{AAE05} and Christodoulou and Koutsoupias~\cite{CK05} proved that the price of anarchy of congestion
games is $5/2$ for linear latency functions and $d^{\Theta(d)}$ for polynomial latency functions of degree $d$. Subsequently, Aland {\em et al.} \cite{ADGMS06} obtained exact bounds on the price of anarchy for congestion games with polynomial latency functions. Still for linear latencies, Caragiannis {\em et al.} \cite{CFKKM06} proved that the same bounds hold for {\em load balancing} games as well, that is, for the restriction in which all possible strategies are singleton sets, while for {\em symmetric} load balancing games, that is load balancing games in which the players share the same set of strategies, L\"{u}cking {\em et al.} \cite{LMMR08} proved that the price of anarchy is $4/3$. Moreover, the works of Caragiannis {\em et al.}~\cite{CFKKM06} and Christodoulou and Koutsoupias~\cite{CK05b} show that the price of stability of congestion games with linear latency functions is $1+1/\sqrt{3}$, while an exact characterization for the case of polynomial latency functions of degree $d$ has been recently given by Christodoulou and Gairing \cite{CG13}.

\

\noindent{\bf Motivations and Previous Related Works.} To the best of our knowledge, Chen and Kempe \cite{CK08} were the first to study the effects of altruistic (and spiteful) behavior on the existence and inefficiency of pure Nash equilibria in some well-understood non-cooperative games. They focus on the class of non-atomic congestion games, where there are infinitely many players each contributing for a negligible amount of congestion, and show that price of anarchy decreases as the degree of altruism of the players increases.

Hoefer and Skopalik \cite{HS09} consider (atomic) linear congestion games with $\gamma_i$-altruistic players, where $\gamma_i\in [0,1]$, for each player $i$. According to their model, player $i$ aims at minimizing a function defined as $1-\gamma_i$ times her cost plus $\gamma_i$ times the sum of the costs of all the players in the game ({\em also counting} player $i$). They show that pure Nash equilibria are always guaranteed to exist via a potential function argument, while, in all the other cases in which existence is not guaranteed, they study the complexity of the problem of deciding whether a pure Nash equilibrium exists in a given game.

Given the existential result by Hoefer and Skopalik \cite{HS09}, Caragiannis {\em et al.}~\cite{CKKKP10} focus on the impact of altruism on the inefficiency of pure Nash equilibria in linear congestion games with altruistic players. However, they consider a more general model of altruistic behavior: in fact, for a parameter $\gamma_i\in [0,1]$, they model a $\gamma_i$-altruistic player $i$ as a player who aims at minimizing a function defined as $1-\gamma_i$ times her cost plus $\gamma_i$ times the sum of the costs of all the players in the game {\em other than} $i$\footnote{Note that each game with $\gamma_i$-altruistic players, where $\gamma_i\in [0,1]$, in the model of Hoefer and Skopalik \cite{HS09} maps to a game with $\gamma'_i$-altruistic players, where $\gamma'_i\in [0,1/2]$, in the model of Caragiannis {\em et al.} \cite{CKKKP10}.}. In such a way, the more $\gamma_i$ increases, the more $\gamma_i$-altruistic players tend to favor the interests of the others to their own ones, with $1$-altruistic and $0$-altruistic players being the two opposite extremal situations in which players behave in a completely altruistic or in a completely selfish way, respectively.
Caragiannis {\em et al.} \cite{CKKKP10} consider the basic case of $\gamma_i=\gamma$ for each player $i$ and show that the price of anarchy is $\frac{5-\gamma}{2-\gamma}$ for $\gamma\in [0,1/2]$ and $\frac{2-\gamma}{1-\gamma}$ for $\gamma\in [1/2,1]$ and that these bounds hold also for load balancing games. This result appears quite surprising, because it shows that altruism can only have a harmful effect on the efficiency of linear congestion games, since the price of anarchy increases from $5/2$ up to an unbounded value as the degree of altruism goes from $0$ to $1$.
On the positive side, they prove that, for the special case of symmetric load balancing games, the price of anarchy is $\frac{4(1-\gamma)}{3-2\gamma}$ for $\gamma\in [0,1/2]$ and $\frac{3-2\gamma}{4(1-\gamma)}$ for $\gamma\in [1/2,1]$, which shows that altruism has a beneficial effect as long as $\gamma\in [0,0.7]$. Note that, that for $\gamma=1/2$, that is when selfishness and altruism are perfectly balanced, the price of anarchy drops to $1$ (i.e., all pure Nash equilibria correspond to socially optimal solutions), while, as soon as $\gamma$ approaches $1$, the price of anarchy again grows up to an unbounded value.

Recently, de Keijzer {\em et al.} \cite{DSAB13} proposed a model for altruistic and spiteful behavior further generalizing the one of Caragiannis {\em et al.} \cite{CKKKP10}. According to their definition, each non-cooperative game with $n$ players is coupled with a real matrix $\Gamma=(\gamma_{ij})\in\R^{n\times n}$, where $\gamma_{ij}$ expresses how much player $i$ cares about player $j$. In such a framework, player $i$ wants to minimize the sum, for each player $j$ in the game (thus also counting $i$), of the cost of player $j$ multiplied by $\gamma_{ij}$. Thus, a positive (resp. negative) value $\gamma_{ij}$ expresses an altruistic (resp. spiteful) attitude of player $i$ towards player $j$. When considering linear congestion games with altruistic players, along the lines of the negative results of Caragiannis {\em et al.} \cite{CKKKP10}, as soon as there are two players $i,j$ such that $\gamma_{ij}>\gamma_{ii}$, i.e., player $i$ cares more about player $j$ than about herself, the price of anarchy becomes unbounded. Therefore, Keijzer {\em et al.} \cite{DSAB13} focus on the scenario, which they call {\em restricted altruistic social context}, in which $\gamma_{ii}\geq\gamma_{ij}$ for each pair of players $i$ and $j$. By extending the smoothness framework of Roughgarden \cite{R12}, they show an upper bound of $7$ on the price of anarchy of coarse correlated equilibria, which implies the same upper bound also on the price of anarchy of correlated equilibria, mixed Nash equilibria and pure Nash equilibria (whenever the latter exist). Moreover, they prove that, when restricting to load balancing games with identical resources, such an upper bound decreases to $2+\sqrt{5}\approx 4.236$.

Noting that matrix $\Gamma$ implicitly represents the {\em social context} (for instance, a social network) in which the players operate, the model of de Keijzer {\em et al.} \cite{DSAB13} falls within the scope of the so-called {\em social context games}. In these games, the payoff of each player is redefined as a function, called {\em aggregating function}, of her cost and of those of her neighbors in a given {\em social context graph}.

Social context games have been introduced and studied by Ashlagi, Krysta, and Tennenholtz \cite{AKT08} for the class of load balancing games, in the case in which the aggregating function is one among the minimum, maximum, sum and ranking functions, for which they gave an almost complete characterization of the cases in which existence of pure Nash equilibria is guaranteed. The model of de Keijzer {\em et al.} \cite{DSAB13}, hence, coincides with a social context game in which the social context graph has weighted edges and the aggregating function is a weighted sum.
The issues of existence and inefficiency of pure Nash equilibria for the case of social context linear congestion games have been considered by Bil\`o {\em et al.}~\cite{BCFG13}. In particular, for the aggregating function sum, pure Nash equilibria are shown to exist for each social context graph via an exact potential function argument and the price of anarchy is shown to fall within the interval $[5;17/3]$.

Finally, the particular case of social context games in which the social context graph is a partition into cliques coincide with games in which static coalitions among players are allowed. These games have been considered by Fotakis, Kontogiannis and Spirakis \cite{FKS06} who focus on weighted congestion game defined on a parallel link graph when the aggregating function is the maximum function (i.e, the coalitional generalization of the KP-model of Koutsoupias and Papadimitriou \cite{KP99}). Among their findings, they show that such games always admit a potential function which becomes an exact one in case of linear latency functions (even in the generalization to networks) and that the price of anarchy is $\Theta\left(\min\left\{k,\frac{\log m}{\log\log m}\right\}\right)$, where $m$ denotes the number of links and $k$ denotes the number of coalitions.

\

\noindent{\bf Our Contribution.} We consider the issues of existence and inefficiency of pure Nash equilibria in linear congestion games with social context as defined by de Keijzer {\em et al.} \cite{DSAB13}. In particular, we restrict our attention to the case of altruistic players, that is, the case in which the matrix $\Gamma$ has only non-negative entries. We show that pure Nash equilibria are always guaranteed to exist via an exact potential function argument, when either the altruistic social context is restricted and $\Gamma$ is symmetric. Moreover, we provide instances with three players not admitting pure Nash equilibria as soon as exactly one of these two properties is not satisfied.

We then prove that, in the restricted altruistic social context, the price of anarchy is exactly $17/3$. Such a characterization is achieved by providing an upper bound of $17/3$ which holds for any matrix $\Gamma$ (even the ones for which pure Nash equilibria are not guaranteed to exist) and a matching lower bound which holds even in the special case in which $\Gamma$ is a boolean symmetric matrix and the game is a load balancing one. Such a result has two interesting interpretations: first, it shows that either the upper bound of $7$ given by de Keijzer {\em et al.} \cite{DSAB13} for the price of anarchy of coarse correlated equilibria is not tight, or that the prices of anarchy of coarse correlated equilibria and pure Nash equilibria are different (the latter hypothesis would be an interesting one, since this situation does not happen in linear congestion games with selfish players); secondly, it proves that the assumption of having identical resources is essential in the upper bound of $2+\sqrt{5}$ given by de Keijzer {\em et al.} \cite{DSAB13} for the case of load balancing games.

For the price of stability in the restricted altruistic social context, we give an upper bound of $2$ holding for each symmetric matrix $\Gamma$ and a lower bound of $1+1/\sqrt{2}\approx 1.707$ holding for the case in which $\Gamma$ is a boolean symmetric matrix.

Finally, we also consider the special case in which $\Gamma$ is such that $\gamma_{ij}=\gamma_i$ for each pair of indexes $i,j$ with $i\neq j$, which coincides with the general model of $\gamma_i$-altruistic players of Caragiannis {\em et al.} \cite{CKKKP10}. We show that pure Nash equilibria are always guaranteed to exist in any case via an exact potential function argument (this slightly improves the existential result by Hoefer and Skopalik \cite{HS09} since they only proved the existence of a weighted potential function) and give an upper bound on the price of anarchy in the general case and an exact bound on the price of stability when $\gamma_i=\gamma$ for each player $i$.

%

\section{Preliminaries}
A congestion game is a tuple ${\cal G}=\langle [n],E,S_{i\in [n]},\ell_{e\in E}\rangle$, where $[n]:=\{1,\ldots,n\}$ is a set of $n\geq 2$ players, $E$ is a set of resources, $\emptyset\neq S_i\subseteq 2^E$ is the set of strategies of player $i$, and $\ell_e:\N\rightarrow\RP$ is the latency function of resource $e$. The special case in which, for each $i\in [n]$ and each $s\in S_i$, it holds $|s|=1$ is called load balancing congestion game. Denoted by ${\cal S}:=\times_{i\in [n]}S_i$ the set of strategy profiles in $\cal G$, that is, the set of outcomes of $\cal G$ in which each player selects a single strategy, the cost of player $i$ in the strategy profile $S=(s_1,\ldots,s_n)\in{\cal S}$ is defined as $c_i(S)=\sum_{e\in s_i}\ell_e(n_e(S))$, where $n_e(S):=|\{j\in [n]:e\in s_j\}|$ is the congestion of resource $e$ in $S$, that is, the number of players using $e$ in $S$.

Given a strategy profile $S=(s_1,\ldots,s_n)$ and a strategy $t\in S_i$ for a player $i\in [n]$, we denote with $S_{-i}\diamond t$ the strategy profile obtained from $S$ by replacing the strategy played by $i$ in $S$ with $t$. A pure Nash equilibrium is a strategy profile $S$ such that, for any player $i\in [n]$ and for any strategy $t\in S_i$, it holds $c_i(S_{-i}\diamond t)\geq c_i(S)$.

The function $\Sum:{\cal S}\rightarrow\RP$ such that $\Sum(S)=\sum_{i\in [n]}c_i(S)$, called the social function, measures the social welfare of a game. Given a congestion game $\cal G$, let ${\cal NE}({\cal G})$ denote the set of its pure Nash equilibria (such a set has been shown to be non-empty by Rosenthal \cite{R73}) and $S^*$ be the strategy profile minimizing the social function. The price of anarchy (\poa) of $\cal G$ is defined as $\max_{S\in{\cal NE(G)}}\left\{\frac{\Sum(S)}{\Sum(S^*)}\right\}$, while the price of stability (\pos) of $\cal G$ is defined as $\min_{S\in{\cal NE(G)}}\left\{\frac{\Sum(S)}{\Sum(S^*)}\right\}$.

A linear congestion game is a congestion game such that, for each $e\in E$, it holds $\ell_e(x)=\alpha_e x+\beta_e$, with $\alpha_e,\beta_e\geq 0$. For these games, the cost of player $i$ in the strategy profile $S=(s_1,\ldots,s_n)$ becomes $c_i(S)=\sum_{e\in s_i}\left(\alpha_e n_e(S)+\beta_e\right)$, while the social value of $S$ becomes $\Sum(S)=\sum_{i\in [n]}\sum_{e\in s_i}\left(\alpha_e n_e(S)+\beta_e\right)=\sum_{e\in E}\left(\alpha_e n_e(S)^2+\beta_e n_e(S)\right)$.

A linear congestion game with an altruistic social context is a pair $({\cal G},\Gamma)$ such that $\cal G$ is a linear congestion game with $n$ players and $\Gamma=(\gamma_{ij})\in\R^{n\times n}$ is a real matrix such that $\gamma_{ij}\geq 0$ for each $i,j\in [n]$. The set of players and strategies is defined as in the underlying linear congestion game $\cal G$, while, for any strategy profile $S$, the cost of player $i$ is $S=(s_1,\ldots,s_n)$ is defined as
$\widehat{c}_i(S)=\sum_{j\in [n]}\left(\gamma_{ij}\cdot c_j(S)\right)=\sum_{j\in [n]}\left(\gamma_{ij}\left(\alpha_e n_e(S)+\beta_e\right)\right)=\sum_{e\in E}\left(\left(\alpha_e n_e(S)+\beta_e\right)\sum_{j\in [n]:e\in s_j}\gamma_{ij}\right)$,
where $c_j(S)$ is the cost of player $j$ in $S$ in the underlying linear congestion game $\cal G$. The special case in which $\gamma_{ii}\geq\gamma_{ij}$ for each $i,j\in [n]$, is called restricted altruistic social context. Note that, in such a case, as pointed out by de Keijzer {\em et al.} \cite{DSAB13}, it is possible to assume without loss of generality that $\gamma_{ii}=1$ for each $i\in [n]$\footnote{This claim follows from the fact that both the set of pure Nash equilibria and the social value of any strategy profile do not change when dividing all the entries in row $i$ of $\Gamma$ by the value $\gamma_{ii}$.}.

\section{Existence of Pure Nash Equilibria}\label{sec-existence}
In this section, we provide a complete characterization of the social contexts for which pure Nash equilibria are guaranteed to exist, independently of which is the underlying linear congestion game.

For a strategy profile $S=(s_1,\ldots,s_n)$, a player $i\in [n]$ and a strategy $t\in S_i$, for the sake of brevity, let us denote with $x_e:=n_e(S)$ and with $x'_e=n_e(S_{-i}\diamond t)$. It holds

\begin{equation}\label{eq-deviation}
\begin{split}
& \widehat{c}_i(S)-\widehat{c}_i(S_{-i}\diamond t)\\
= & \displaystyle\sum_{e\in E}\left(\left(\alpha_e x_e+\beta_e\right)\sum_{j:e\in s_j}\gamma_{ij}\right)-\sum_{e\in E}\left(\left(\alpha_e x'_e+\beta_e\right)\sum_{j:e\in s_j}\gamma_{ij}\right)\\
= & \displaystyle\sum_{e\in s_i\setminus t}\left(\left(\alpha_e x_e+\beta_e\right)\sum_{j:e\in s_j}\gamma_{ij}-\left(\alpha_e (x_e-1)+\beta_e\right)\sum_{j\neq i:e\in s_j}\gamma_{ij}\right)\\
& +\displaystyle\sum_{e\in t\setminus s_i}\left(\left(\alpha_e x_e+\beta_e\right)\sum_{j:e\in s_j}\gamma_{ij}-\left(\alpha_e (x_e+1)+\beta_e\right)\left(\gamma_{ii}+\sum_{j:e\in s_j}\gamma_{ij}\right)\right)\\
= & \displaystyle\sum_{e\in s_i\setminus t}\left(\gamma_{ii}\left(\alpha_e x_e+\beta_e\right)+\alpha_e\sum_{j\neq i:e\in s_j}\gamma_{ij}\right)\\
& -\displaystyle\sum_{e\in t\setminus s_i}\left(\gamma_{ii}\left(\alpha_e (x_e+1)+\beta_e\right)+\alpha_e\sum_{j:e\in s_j}\gamma_{ij}\right).
\end{split}
\end{equation}

On the positive side, we show that, for restricted altruistic social contexts such that $\Gamma$ is symmetric, pure Nash equilibria do always exist.

\begin{theorem}\label{existence}
Each linear congestion game with restricted altruistic social context $({\cal G},\Gamma)$ such that $\Gamma$ is symmetric admits an exact potential function.
\end{theorem}
\begin{proof}
For a strategy profile $S$ and a resource $e$, let $P_e(S)=\{(i,j)\in [n]\times [n]:i\neq j \wedge e\in s_i\cap s_j\}$. We define the following potential function: $$\Phi(S)=\frac 1 2 \sum_{e\in E}\left(\alpha_e\left(n_e(S)(n_e(S)+1)+\sum_{(i,j)\in P_e(S)}\gamma_{ij}\right)+2\beta_e n_e(S)\right).$$ Consider a strategy profile $S=(s_1,\ldots,s_n)$, a player $i\in [n]$ and a strategy $t\in S_i$ and again denote with $x_e:=n_e(S)$. For the case in which $\gamma_{ii}=1$ for each $i\in [n]$ and $\gamma_{ij}=\gamma_{ji}$ for each $i,j\in [n]$, it holds
\begin{displaymath}
\begin{array}{cl}
& \Phi(S)-\Phi(S_{-i}\diamond t)\\
= & \displaystyle\frac 1 2 \sum_{e\in s_i\setminus t}\left(\alpha_e\left(x_e(x_e+1)+\sum_{(i,j)\in P_e(S)}\gamma_{ij}\right)+2\beta_e x_e\right)\\
& -\displaystyle\frac 1 2 \sum_{e\in s_i\setminus t}\left(\alpha_e\left(x_e(x_e-1)+\sum_{(j,k)\in P_e(S):j,k\neq i}\gamma_{jk}\right)+2\beta_e(x_e-1)\right)\\
& +\displaystyle\frac 1 2 \sum_{e\in t\setminus s_i}\left(\alpha_e\left(x_e(x_e+1)+\sum_{(i,j)\in P_e(S)}\gamma_{ij}\right)+2\beta_e x_e\right)\\
& -\displaystyle\frac 1 2 \sum_{e\in t\setminus s_i}\left(\alpha_e\left((x_e+1)(x_e+2)+\sum_{(j,k)\in P_e(S)}\gamma_{jk}+2\sum_{j:e\in s_j}\gamma_{ij}\right)+2\beta_e(x_e+1)\right)\\
= & \displaystyle\frac 1 2 \sum_{e\in s_i\setminus t}\left(\alpha_e\left(2x_e+2\sum_{j\neq i:e\in s_j}\gamma_{ij}\right)+2\beta_e\right)\\
& -\displaystyle\frac 1 2 \sum_{e\in t\setminus s_i}\left(\alpha_e\left(2(x_e+1)+2\sum_{j:e\in s_j}\gamma_{ij}\right)+2\beta_e\right)\\
= & \displaystyle\sum_{e\in s_i\setminus t}\left(\alpha_e\left(x_e+\sum_{j\neq i:e\in s_j}\gamma_{ij}\right)+\beta_e\right)-\sum_{e\in t\setminus s_i}\left(\alpha_e\left(x_e+1+\sum_{j:e\in s_j}\gamma_{ij}\right)+\beta_e\right)\\
\end{array}
\end{displaymath}
which, being equivalent to equation (\ref{eq-deviation}), shows that $\Phi$ is an exact potential function for $({\cal G},\Gamma)$.\qed
\end{proof}

In order to prove that the characterization given in Theorem \ref{existence} is tight, we provide the following two non-existential results. In the first one, although preserving the property that $\Gamma$ is symmetric, we relax the constraint that the game is played in a restricted altruistic social context: in particular, we allow $\gamma_{ii}=0$ for some player $i\in [n]$.

\begin{theorem}\label{ne1}
There exists a three-player linear congestion game $\cal G$ and a symmetric matrix $\Gamma\in\R^{3\times 3}$ such that the linear congestion game with altruistic social context $({\cal G},\Gamma)$ does not admit pure Nash equilibria.
\end{theorem}

In the second result, although preserving the property that the game is played in a restricted altruistic social context, we relax the constraint that $\Gamma$ is symmetric.

\begin{theorem}\label{ne2}
There exists a three-player linear congestion game $\cal G$ and a matrix $\Gamma\in\R^{3\times 3}$ with a unitary main diagonal such that $({\cal G},\Gamma)$ does not admit pure Nash equilibria.
\end{theorem}

\section{Inefficiency of Pure Nash Equilibria}\label{sec-bounds}
In this section, we give bounds on the prices of anarchy and stability of linear congestion games with restricted social context. These bounds are achieved by applying the primal-dual technique that we introduced in \cite{B12}. To this aim, we recall that it is possible to assume without loss of generality that $\beta_e=0$ for each $e\in E$ as long as we are not interested in load balancing games. For a given linear congestion game with altruistic social context $({\cal G},\Gamma)$, we denote with $K=(k_1,\ldots,k_n)$ and $O=(o_1,\ldots,o_n)$, respectively, a Nash equilibrium and a social optimum of $({\cal G},\Gamma)$ and we use $K_e:=n_e(K)$ and $O_e:=n_e(O)$ to denote the congestion of resource $e$ in $K$ and $O$, respectively.

The primal-dual method aims at formulating the problem of maximizing the ratio $\frac{{\sf SUM}(K)}{{\sf SUM}(O)}$ via linear programming. The two strategy profiles $K$ and $O$ play the role of fixed constants, while, for each $e\in E$, the values $\alpha_e$ defining the latency functions are variables that must be suitably chosen so as to satisfy two constraints: the first, assures that $K$ is a pure Nash equilibrium, while the second normalizes to $1$ the value of the social optimum ${\sf SUM}(O)$. The objective function aims at maximizing the social value ${\sf SUM}(K)$ which, being the social optimum normalized to $1$, is equivalent to maximize the ratio $\frac{{\sf SUM}(K)}{{\sf SUM}(O)}$. Let us denote with $LP(K,O)$ such a linear program, which, in our scenario of investigation becomes
\begin{displaymath}
\begin{array}{ll}
maximize \displaystyle\sum_{e\in E}\left(\alpha_e K_e^2\right)\\
subject\ to\\
\displaystyle\sum_{e\in k_i\setminus o_i}\left(\alpha_e \left(K_e+\sum_{j\neq i:e\in k_j}\gamma_{ij}\right)\right)\\
\ \ -\displaystyle\sum_{e\in o_i\setminus k_i}\left(\alpha_e\left(K_e+1+\sum_{j:e\in k_j}\gamma_{ij}\right)\right)\leq 0, & \ \ \forall i\in [n]\\
\displaystyle\sum_{e\in E}\left(\alpha_e O_e^2\right) = 1,\\
\alpha_e\geq 0, & \ \ \forall e\in E
\end{array}
\end{displaymath}
Let $DLP(K,O)$ be the dual program of $LP(K,O)$. By the Weak Duality Theorem, each feasible solution to $DLP(K,O)$ provides an upper bound on the optimal solution of $LP(K,O)$. Hence, by providing a feasible dual solution, we obtain an upper bound on the ratio $\frac{{\sf SUM}(K)}{{\sf SUM}(O)}$. Anyway, if the provided dual solution is independent on the particular choice of $K$ and $O$, we obtain an upper bound on the ratio $\frac{{\sf SUM}(K)}{{\sf SUM}(O)}$ for any possible pair of profiles $K$ and $O$, which means that we obtain an upper bound on the price of anarchy of pure Nash equilibria.
The dual program $DLP(K,O)$ is
\begin{displaymath}
\begin{array}{ll}
minimize\ \theta\\
subject\ to\\
\displaystyle\sum_{i:e\in k_i\setminus o_i}\left(y_i \left(K_e+\sum_{j\neq i:e\in k_j}\gamma_{ij}\right)\right)\\
\ \ -\displaystyle\sum_{i:e\in o_i\setminus k_i}\left(y_i \left(K_e+1+\sum_{j:e\in k_j}\gamma_{ij}\right)\right)+\theta O_e^2 \geq K_e^2, & \ \ \forall e\in E\\
y_i\geq 0, & \ \ \forall i\in [n]
\end{array}
\end{displaymath}

\begin{theorem}
For any linear congestion game with restricted altruistic social context $({\cal G},\Gamma)$, it holds ${\sf PoA}({\cal G},\Gamma)\leq\frac{17}{3}$.
\end{theorem}
\begin{proof}
Consider the dual solution such that $\theta=17/3$ and $y_i=5/3$ for each $i\in [n]$. With these values, for each $e\in E$, the dual constraint becomes
$$5\sum_{i:e\in k_i\setminus o_i}\left(K_e+\sum_{j\neq i:e\in k_j}\gamma_{ij}\right)-5\sum_{i:e\in o_i\setminus k_i}\left(K_e+1+\sum_{j:e\in k_j}\gamma_{ij}\right)+17O_e^2 \geq 3K_e^2.$$
Let $\Delta_e=|\{i\in [n]:e\in k_i\cap o_i\}|$. Since $({\cal G},\Gamma)$ is a linear congestion game with restricted altruistic social context, it holds $\gamma_{ij}\leq 1$ for each $i,j\in [n]$. Hence, the dual constraint is obviously verified when it holds
$$5\sum_{i:e\in k_i\setminus o_i}K_e-5\sum_{i:e\in o_i\setminus k_i}\left(2K_e+1\right)+17O_e^2 \geq 3K_e^2,$$
which is equivalent to
\begin{equation}\label{eq3}
5(K_e-\Delta_e)K_e-5(O_e-\Delta_e)\left(2K_e+1\right)+17O_e^2 \geq 3K_e^2.
\end{equation}
It is easy to see that inequality (\ref{eq3}) is always true when it holds
\begin{equation}\label{eq4}
5K^2_e-5O_e\left(2K_e+1\right)+17O_e^2 \geq 3K_e^2.
\end{equation}
To see that inequality (\ref{eq4}) is always verified for any pair of non-negative integers $(K_e,O_e)$, note that the discriminant of its associate equality, when solved for $K_e$, is non-positive for any $O_e\geq 2$ and that inequality (\ref{eq4}) is always verified for any non-negative values of $K_e$ when $O_e\in\{0,1\}$.\qed
\end{proof}

When compared to the upper bound of $7$ for the price of anarchy of coarse correlated equilibria given by de Keijzer {\em et al.} \cite{DSAB13}, our upper bound implies that either the one for coarse correlated equilibria is not tight, or the prices of anarchy of coarse correlated equilibria and pure Nash equilibria are different. Such a latter case would be quite significant since this does not happen in linear congestion games with selfish players.

We now give a marching lower bound which holds even in the special case in which $\Gamma$ is a symmetric boolean matrix and the underlying linear congestion game is a load balancing one. The basic idea of our construction, suitably extended to comply with our altruistic scenario, is borrowed from Caragiannis {\em et al.} \cite{CFKKM06}.

\begin{theorem}\label{lbpoa}
For any $\epsilon>0$, there exists a linear congestion game with restricted altruistic social context $({\cal G},\Gamma)$, such that $\cal G$ is a load balancing game and $\Gamma$ is a symmetric boolean matrix, for which ${\sf PoA}({\cal G},\Gamma)\geq\frac{17}{3}-\epsilon$.
\end{theorem}

Note that such a lower bound implies that the the assumption of identical resources in crucial in the upper bound of $2+\sqrt{5}$ given by de Keijzer {\em et al.} \cite{DSAB13} for load balancing games with restricted altruistic social context.

We now turn our attention to the study of the price of stability. By exploiting the potential function defined in the previous section and the fact that there exists a pure Nash equilibrium $K$ such that $\Phi(K)\leq\Phi(O)$, we easily obtain the following upper bound.

\begin{theorem}
For any linear congestion game with restricted altruistic social context $({\cal G},\Gamma)$ such that $\Gamma$ is symmetric, it holds ${\sf PoS}({\cal G},\Gamma)\leq 2$.
\end{theorem}
\begin{proof}
Let $K$ be a pure Nash equilibrium obtained after a sequence of improving deviation starting from $O$. The existence of $K$ is guaranteed by the existence of the potential function $\Phi$. Moreover, it holds $\Phi(K)\leq\Phi(O)$. Hence, it follows that
\begin{eqnarray*}
{\sf SUM}(K) & = & \sum_{e\in E}\left(\alpha_e K_e^2\right)\\
& \leq & \sum_{e\in E}\left(\alpha_e\left(K_e(K_e+1)+\sum_{(i,j)\in P_e(K)}\gamma_{ij}\right)\right)\\
& = & \Phi(K)\\
& \leq & \Phi(O)\\
& = & \sum_{e\in E}\left(\alpha_e\left(O_e(O_e+1)+\sum_{(i,j)\in P_e(O)}\gamma_{ij}\right)\right)\\
& \leq & \sum_{e\in E}\left(\alpha_e\left(O_e(O_e+1)+O_e(O_e-1)\right)\right)\\
& = & 2\sum_{e\in E}\left(\alpha_e O_e^2\right)\\
& = & 2{\sf SUM}(O),
\end{eqnarray*}
where the last inequality follows from the fact that $\gamma_{ij}\in [0,1]$ for each $i,j\in [n]$ and $|P_e(O)|=O_e(O_e-1)$.\qed
\end{proof}

In this case, we are only able to provide a lower bound of $1+\frac{1}{\sqrt{2}}\approx 1.707$.

\begin{theorem}\label{lbpos}
For any $\epsilon>0$, there exists a linear congestion game with restricted altruistic social context $({\cal G},\Gamma)$, such that $\Gamma$ is a symmetric boolean matrix, for which ${\sf PoS}({\cal G},\Gamma)\geq 1+\frac{1}{\sqrt{2}}-\epsilon$.
\end{theorem}

\section{Results for Simple Social Contexts}\label{sec-special}

In this section, we focus on the special case given by model of Caragiannis {\em et al.}~\cite{CKKKP10} in which, for each $i\in [n]$, it holds $\gamma_{ii}=1-\gamma_i$ and $\gamma_{ij}=\gamma_i$ for each $j\neq i\in [n]$, where $\gamma_i\in [0,1]$. In such a model, the restricted altruistic social context coincides with the case in which, for each $i\in [n]$, it holds $\gamma_i\leq 1/2$. Caragiannis {\em et al.} \cite{CKKKP10} show that, when $\gamma_i=\gamma$ for each $i\in [n]$, the price of anarchy is exactly $\frac{2-\gamma}{1-\gamma}$ for general altruistic social contexts and $\frac{5-\gamma}{2-\gamma}$ in the restricted one.

First of all, we prove that pure Nash equilibria are always guaranteed to exist via an exact potential function argument. An existential result had already been given by Hoefer and Skopalik \cite{HS09}, nevertheless, their proof makes use of a weighted potential function. So, our result is slightly stronger and, more importantly, provides a better potential function to be subsequently exploited in the derivation of an upper bound on the price of stability of these games.

Let ${\cal V}_n:=[0,1]^n$ be the set of $n$-dimensional vectors whose entries belong to the interval $[0,1]$. Given a vector $V=(v_1,\ldots,v_n)\in{\cal V}_n$, denote with $\Gamma_V$ the $n\times n$ matrix $\Gamma$ such that, for each $i\in [n]$, it holds $\gamma_{ii}=1-v_i$ and $\gamma_{ij}=v_i$ for each $j\neq i\in [n]$.

\begin{theorem}
Each $n$-player linear congestion game with altruistic social context $({\cal G},\Gamma)$ such that $\Gamma=\Gamma_V$ for some $V\in{\cal V}_n$ admits an exact potential function.
\end{theorem}
\begin{proof}
Consider a strategy profile $S=(s_1,\ldots,s_n)$, a player $i\in [n]$ and a strategy $t\in S_i$, and again denote with $x_e:=n_e(S)$. From equation (\ref{eq-deviation}), since $\sum_{j\neq i:e\in s_j}\gamma_{ij}=(x_e-1)v_i$ and $\sum_{j:e\in s_j}\gamma_{ij}=x_e v_i$, it follows that
\begin{eqnarray}\label{eq-dev-etero}
\begin{split}
& \widehat{c}_i(S)-\widehat{c}_i(S_{-i}\diamond t)\\
= & \sum_{e\in s_i\setminus t}\left(\alpha_e\left(x_e-v_i\right)+(1-v_i)\beta_e\right)-\sum_{e\in t\setminus s_i}\left(\alpha_e\left(x_e+1-v_i\right)+(1-v_i)\beta_e\right).
\end{split}
\end{eqnarray}
Consider, now, the following potential function $$\Phi(S)=\frac 1 2\sum_{e\in E}\left(\alpha_e\left(x_e(x_e+1)-2\sum_{j:e\in s_j}v_j\right)+2\beta_e \sum_{j:e\in s_j}(1-v_j)\right).$$
It holds
\begin{displaymath}
\begin{array}{cl}
& \Phi(S)-\Phi(S_{-i}\diamond t)\\
= & \displaystyle\frac 1 2 \sum_{e\in s_i\setminus t}\left(\alpha_e\left(2x_e-2v_i\right)+2(1-v_i)\beta_e\right)\\
& -\displaystyle\frac 1 2 \sum_{e\in t\setminus s_i}\left(\alpha_e\left(2(x_e+1)-2v_i\right)+2(1-v_i)\beta_e\right)\\
= & \displaystyle\sum_{e\in s_i\setminus t}\left(\alpha_e\left(x_e-v_i\right)+(1-v_i)\beta_e\right)-\sum_{e\in t\setminus s_i}\left(\alpha_e\left(x_e+1-v_i\right)+(1-v_i)\beta_e\right)\\
\end{array}
\end{displaymath}
which shows that $\Phi$ is an exact potential function for $({\cal G},\Gamma)$.\qed
\end{proof}

By exploiting the potential function defined above, we obtain an upper bound on the price of stability for the case in which $v_i=v$ for each $i\in [n]$ as follows.

The fact that there exists a pure Nash equilibrium $K$ such that $\Phi(K)\leq\Phi(O)$ easily implies the following inequality (where, as usual, we have removed the terms $\beta_e$ from the latency functions):
\begin{equation}\label{eqpos}
\sum_{e\in E}\left(\alpha_e\left(K_e(K_e+1)-2vK_e-O_e(O_e+1)+2vO_e\right)\right)\leq 0,
\end{equation}
where we have used the equalities $\sum_{j:e\in k_j}v_j=v K_e$ and $\sum_{j:e\in o_j}v_j=v O_e$.

By exploiting the inequality $\widehat{c}_i(K)-\widehat{c}_i(K_{-i}\diamond o_i)\leq 0$, we obtain that, for each $i\in [n]$, it holds
\begin{equation*}
\sum_{e\in k_i\setminus o_i}\left(\alpha_e\left(K_e-v\right)\right)-\sum_{e\in o_i\setminus k_i}\left(\alpha_e\left(K_e+1-v\right)\right)\leq 0,
\end{equation*}
which implies
\begin{equation}\label{eqnash}
\sum_{e\in k_i}\left(\alpha_e\left(K_e-v\right)\right)-\sum_{e\in o_i}\left(\alpha_e\left(K_e+1-v\right)\right)\leq 0.
\end{equation}

Using both inequalities (\ref{eqpos}) and (\ref{eqnash}), the primal formulation $LP(K,O)$ becomes the following one.

\begin{displaymath}
\begin{array}{ll}
maximize \displaystyle\sum_{e\in E}\left(\alpha_e K_e^2\right)\\\vspace{0.1cm}
subject\ to\\\vspace{0.1cm}
\displaystyle\sum_{e\in E}\left(\alpha_e\left(K_e(K_e+1)-2vK_e-O_e(O_e+1)+2vO_e\right)\right)\leq 0\\\vspace{0.1cm}
\displaystyle\sum_{e\in k_i}\left(\alpha_e \left(K_e-v\right)\right)-\sum_{e\in o_i}\left(\alpha_e\left(K_e+1-v\right)\right)\leq 0, & \ \ \forall i\in [n]\\\vspace{0.1cm}
\displaystyle\sum_{e\in E}\left(\alpha_e O_e^2\right) = 1,\\\vspace{0.1cm}
\alpha_e\geq 0, & \ \ \forall e\in E
\end{array}
\end{displaymath}
The dual program $DLP(K,O)$ is
\begin{displaymath}
\begin{array}{ll}
minimize\ \theta\\\vspace{0.1cm}
subject\ to\\\vspace{0.1cm}
x\left(K_e(K_e+1)-2vK_e-O_e(O_e+1)+2vO_e\right)\\
\ \ +\displaystyle\sum_{i:e\in k_i}\left(y_i \left(K_e-v\right)\right)-\sum_{i:e\in o_i}\left(y_i \left(K_e+1-v\right)\right)+\theta O_e^2 \geq K_e^2, & \ \ \forall e\in E\\\vspace{0.1cm}
x\geq 0,\\\vspace{0.1cm}
y_i\geq 0, & \ \ \forall i\in [n]
\end{array}
\end{displaymath}

\begin{theorem}\label{ubposetero}
For any $n$-player linear congestion game with altruistic social context $({\cal G},\Gamma)$ such that $\Gamma=\Gamma_V$ for some $V\in{\cal V}_n$ with $v_i=v$ for each $i\in [n]$, it holds ${\sf PoS}({\cal G},\Gamma)\leq\frac{(\sqrt{3}+1)(1-v)}{\sqrt{3}-v(\sqrt{3}-1)}$ when $v\in [0,1/2]$ and ${\sf PoS}({\cal G},\Gamma)\leq\frac{3-\sqrt{3}-2v(2-\sqrt{3})}{2(1-v)}$ when $v\in [1/2,1]$.
\end{theorem}

We now show matching lower bounds.

\begin{theorem}\label{lbposspecial}
For any $\epsilon>0$, there exists an $n$-player linear congestion game with altruistic social context $({\cal G},\Gamma)$ such that $\Gamma=\Gamma_V$ for some $V\in{\cal V}_n$ with $v_i=v\in [0,1/2]$ for each $i\in [n]$ for which it holds ${\sf PoS}({\cal G},\Gamma)\geq\frac{(\sqrt{3}+1)(1-v)}{\sqrt{3}-v(\sqrt{3}-1)}-\epsilon$ and an $n$-player linear congestion game with altruistic social context $({\cal G}',\Gamma')$ such that $\Gamma'=\Gamma'_V$ for some $V\in{\cal V}_n$ with $v_i=v\in [1/2,1]$ for each $i\in [n]$ for which it holds ${\sf PoS}({\cal G}',\Gamma')\geq\frac{3-\sqrt{3}-2v(2-\sqrt{3})}{2(1-v)}-\epsilon$.
\end{theorem}

Note that for $v=1/2$, the price of stability is $1$ which means that, when players are half selfish and half altruistic, there always exists a social optimal solution which is also a pure Nash equilibrium. For $v=0$, that is, when players are totally selfish, we reobtain the well-known bound of $1+1/\sqrt{3}$ on the price of stability of linear congestion games proven by Caragiannis {\em et al.} \cite{CFKKM06}. For $v=1$, the price of stability goes to infinity, i.e., all Nash equilibria may perform extremely bad with respect to the social optimal solution. This implies that totally altruistic players are tremendously harmful in a non-cooperative system, since they yield games in which even the price of stability may be unbounded. Finally, in the restricted altruistic social context, i.e. $v\in [0,1/2]$, when $v$ goes from $0$ to $1/2$, the price of anarchy increases from $5/2$ to $3$, while the price of stability decreases from $1+1/\sqrt{3}$ to $1$. In particular, the increase in the price of anarchy is always compensated by a slightly higher decrease in the price of stability.

\

For a vector $V\in {\cal V}_n$, denote with $\overline{v}$ and $\underline{v}$ the maximum and minimum entry in $V$, respectively. For the price of anarchy, by simply exploiting inequality (\ref{eq-dev-etero}), we get the following dual program

\begin{displaymath}
\begin{array}{ll}
minimize\ \theta\\\vspace{0.1cm}
subject\ to\\\vspace{0.1cm}
\displaystyle\sum_{i:e\in k_i}\left(x_i \left(K_e-\overline{v}\right)\right)-\sum_{i:e\in o_i}\left(x_i \left(K_e+1-\underline{v}\right)\right)+\theta O_e^2 \geq K_e^2, & \ \ \forall e\in E\\\vspace{0.1cm}
x_i\geq 0, & \ \ \forall i\in [n]
\end{array}
\end{displaymath}

\begin{theorem}\label{ubpoaspecial}
For any $n$-player linear congestion game with altruistic social context $({\cal G},\Gamma)$ such that $\Gamma=\Gamma_V$ for some $V\in{\cal V}_n$, it holds ${\sf PoA}({\cal G},\Gamma)\leq\frac{2-\underline{v}}{1-\overline{v}}$ when $\overline{v}\in [1/2,1]$ and ${\sf PoA}({\cal G},\Gamma)\leq\frac{5+2\overline{v}-3\underline{v}}{2-\overline{v}}$ when $\overline{v}\in [0,1/2]$.
\end{theorem}

Note that, for $\overline{v}=\underline{v}$, we reobtain the upper bounds already proved by Caragiannis {\em et al.} \cite{CKKKP10}. For the general case in which $\overline{v}>\underline{v}$, we have not been able to achieve matching or almost matching lower bounds so far.

\section{Conclusions}
We have focused on the existence and inefficiency of pure Nash equilibria in linear congestion games with altruistic social context in the spirit of the model recently proposed by Keijzer {\em et al.} \cite{DSAB13}.

We have proved that pure Nash equilibria are always guaranteed to exist when the matrix $\Gamma$ defining the social context either has a unitary main diagonal and is symmetric and that this result is tight in the sense that both properties are essential as long as $\Gamma$ does not obey other particular properties. In fact, for the case in which $\Gamma$ is such that $\gamma_{ij}=\gamma_i$ for each $i,j$ with $i\neq j$, existence of pure Nash equilibria can be proved although $\Gamma$ neither has a unitary main diagonal nor is symmetric. Thus, detecting other particular special cases for which such an existential result could be provided is an interesting question.

We have also shown that the price of anarchy for general social contexts is exactly $17/3$ and that this bounds holds even for the case of load balancing games. When compared with the results of Keijzer {\em et al.} \cite{DSAB13}, this gives rise to two important questions. The first one is to determine the exact price of anarchy of coarse correlated equilibria. Is this the same as the one of pure Nash equilibria, or is there a separation result, showing that when going from pure Nash equilibria to coarse correlated equilibria, passing through mixed Nash equilibria and correlated equilibria, at a certain point there must occur an increase in the price of anarchy? The second one, is to determine the exact price of anarchy for the basic case of load balancing games with identical resources.

As to the price of stability, instead, the main issue is to close the gap between our upper bound of $2$ and the lower bound of $1+1/\sqrt{2}$. To this direction, our intuition is that the upper bound is tight.

Finally, for the special case in which $\Gamma$ is such that $\gamma_{ij}=\gamma_i$ for each $i,j$ with $i\neq j$, we have only given upper bounds on the price of anarchy. Matching or nearly matching lower bounds, as well as bounds on the price of stability are still missing.

\section{Appendix}

\subsection{Omitted Material from Section \ref{sec-existence}}

{\bf Proof of Theorem \ref{ne1}.}
Let $\cal G$ be a linear congestion game with $3$ players and $13$ resources such that $S_1=\left\{\{e_1,e_4,e_{13}\},\{e_2,e_3,e_5,e_6\}\right\}$, $S_2=\left\{\{e_4,e_5,e_6,e_9,e_{10},e_{13}\},\{e_3,e_7,e_8\}\right\}$, $S_3=\left\{\{e_1,e_4,e_5,e_7,e_{10},e_{12}\},\{e_6,e_8,e_{11},e_{13}\}\right\}$, $\alpha_1=9$, $\alpha_2=7$, $\alpha_3=16$, $\alpha_4=25$, $\alpha_5=14$, $\alpha_6=32$, $\alpha_7=363$, $\alpha_8=87$, $\alpha_9=383$, $\alpha_{10}=318$, $\alpha_{11}=1047$, $\alpha_{12}=160$, $\alpha_{13}=31$ and $\beta_e=0$ for each $e\in [13]$. The matrix $\Gamma$ is such that $\gamma_{11}=\gamma_{33}=1$, $\gamma_{22}=0$, $\gamma_{12}=\gamma_{21}=\frac{10}{211}$, $\gamma_{13}=\gamma_{31}=\frac{2}{53}$ and $\gamma_{23}=\gamma_{32}=\frac{1}{9}$. Hence, $\Gamma$ is symmetric, but the altruistic social context is not restricted.

It is not difficult to check by inspection that each of the possible eight strategy profiles is not a pure Nash equilibrium. To this aim, we use a triple $(x_1,x_2,x_3)$, with $x_i\in\{1,2\}$ for each $i\in [3]$, to denote the strategy profile in which player $i$ chooses her first or second strategy depending on whether $x_i=1$ or $x_i=2$, respectively. In the following table it is shown that no pure Nash equilibria exist in $({\cal G},\Gamma)$.

\begin{tabular}{|c|c|c|c|c|c|}
  \hline
  \ Profile $S$\  & \ $\widehat{c}_1(S)$\  & \ $\widehat{c}_2(S)$\  & \ $\widehat{c}_3(S)$\  & \ Migrating Player\  & \ New Profile\  \\\hline\hline
  $(1,1,1)$ & $>260$ & - & - & 1 & $(2,1,1)$ \\\hline
  $(2,1,1)$ & $<243$ & $>146$ & $<1398.87$ & 2 & $(2,2,1)$ \\\hline
  $(2,2,1)$ & $>186.82$ & $<146$ & - & 1 & $(1,2,1)$ \\\hline
  $(1,2,1)$ & $<186.82$ & - & $>1381$ & 3 & $(1,2,2)$ \\\hline
  $(1,2,2)$ & - & $>150.66$ & $<1381$ & 2 & $(1,1,2)$ \\\hline
  $(1,1,2)$ & $>244$ & $<150.65$ & - & 1 & $(2,1,2)$ \\\hline
  $(2,1,2)$ & $<239$ & $<151$ & $>1398.88$ & 3 & $(2,1,1)$ \\\hline
  $(2,2,2)$ & - & $>151$ & - & 2 & $(2,1,2)$ \\
  \hline
\end{tabular}

\

\noindent {\bf Proof of Theorem \ref{ne2}.}
Let $\cal G$ be a linear congestion game with $3$ players and $9$ resources such that $S_1=\left\{\{e_3\},\{e_1,e_2\}\right\}$, $S_2=\left\{\{e_3,e_6,e_7\},\{e_2,e_4,e_5\}\right\}$, $S_3=\left\{\{e_3,e_4,e_7,e_9\},\{e_5,e_8\}\right\}$, $\alpha_1=10$, $\alpha_2=1$, $\alpha_3=4$, $\alpha_4=392$, $\alpha_5=98$, $\alpha_6=384$, $\alpha_7=294$, $\alpha_8=1052$, $\alpha_9=160$ and $\beta_e=0$ for each $e\in [9]$. The matrix $\Gamma$ is such that $\gamma_{ij}=1$ for each $i,j\in [3]$ except for $\gamma_{21}=\gamma_{32}=0$. Hence, $({\cal G},\Gamma)$ is a linear congestion game with restricted altruistic social context but $\Gamma$ is not symmetric.

It is not difficult to check by inspection that each of the possible eight strategy profiles is not a pure Nash equilibrium. To this aim, we use a triple $(x_1,x_2,x_3)$, with $x_i\in\{1,2\}$ for each $i\in [3]$, to denote the strategy profile in which player $i$ chooses her first or second strategy depending on whether $x_i=1$ or $x_i=2$, respectively. In the following table it is shown that no pure Nash equilibria exist in $({\cal G},\Gamma)$.

\begin{tabular}{|c|c|c|c|c|c|}
  \hline
  \ Profile $S$\  & \ $\widehat{c}_1(S)$\  & \ $\widehat{c}_2(S)$\  & \ $\widehat{c}_3(S)$\  & \ Migrating Player\  & \ New Profile\  \\\hline\hline
  $(1,1,1)$ & 2148 & - & - & 1 & $(2,1,1)$ \\\hline
  $(2,1,1)$ & 2139 & 2128 & 1159 & 2 & $(2,2,1)$ \\\hline
  $(2,2,1)$ & 2138 & 2126 & - & 1 & $(1,2,1)$ \\\hline
  $(1,2,1)$ & 2137 & - & 1254 & 3 & $(1,2,2)$ \\\hline
  $(1,2,2)$ & - & 1837 & 1252 & 2 & $(1,1,2)$ \\\hline
  $(1,1,2)$ & 1844 & 1836 & - & 1 & $(2,1,2)$ \\\hline
  $(2,1,2)$ & 1843 & 1832 & 1161 & 3 & $(2,1,1)$ \\\hline
  $(2,2,2)$ & - & 1838 & - & 2 & $(2,1,2)$ \\
  \hline
\end{tabular}

\subsection{Omitted Material from Section \ref{sec-bounds}}

\noindent{\bf Proof of Theorem \ref{lbpoa}.}
We use the notion of {\em game graph} introduced by Caragiannis {\em et al.} \cite{CFKKM06} to describe load balancing games in which each player has exactly two possible strategies: the one played at the social optimum and the one played at the worst pure Nash equilibrium. Each node in the graph models a resource, while each edge $\{i,j\}$ corresponds to a player who can only choose between one of the two resources $i$ and $j$.

We define a game graph $T$ consisting in a tree having $2h+1$ levels, numbered from $0$ to $2h$. Level $0$ corresponds to the root and level $2h$ to leaves. Each node at level $i$, with $0\leq i\leq h-1$, has three children, while each node at level $i$, with $h\leq i\leq 2h-1$, has two children. Hence, $T$ is complete ternary tree of height $h$ whose leaves are the roots of complete binary trees of height $h$. This implies that each level $i$, with $0\leq i\leq h$ has $3^i$ nodes, while each level $i$, with $h+1\leq i\leq 2h$, has $3^h 2^{i-h}$ nodes. The latency function of each node at level $i$ is of type $\alpha_i x$, where $\alpha_i=\left(\frac{3}{7}\right)^i$ for each $0\leq i\leq h-1$, $\alpha_i=\frac 3 5\left(\frac{3}{7}\right)^{h-1}\left(\frac{2}{5}\right)^{i-h}$ for each $h\leq i\leq 2h-1$ and $\alpha_{2h}=\frac 6 5\left(\frac{6}{35}\right)^{h-1}$. The matrix $\Gamma$ is defined as follows: each player $j=\{u,v\}$ such that $u$ is the parent of $v$ in $T$ cares of the player corresponding to the edge connecting node $u$ to its parent (we call such a player the {\em parent} of $j$, whenever it exists) and of all the players corresponding to the edges connecting $v$ to its children (we call these players the {\em children} of $j$, whenever they exist). By ``cares", we mean that the corresponding entry in the induced matrix $\Gamma$ is $1$, otherwise, it is $0$. It is easy to see that $\Gamma$ is a symmetric boolean matrix.

We show that the strategy profile $K$ in which each player selects the resource closest to the root is a pure Nash equilibrium for $(T,\Gamma)$.

Consider a player $j$ using a resource $k_j$ belonging to level $i$, with $0\leq i\leq h-2$. Player $j$ is sharing $k_j$ with her two siblings, thus $c_j(K)=3\left(\frac{3}{7}\right)^i$. The three children of player $j$ are sharing the resource $o_j$ belonging to level $i+1$, thus each of them is paying a cost of $3\left(\frac{3}{7}\right)^{i+1}$. Finally, assume that the parent of player $j$ is paying a cost of $\delta$ (with $\delta=0$ when such a player does not exist). It follows that $\widehat{c}_j(K)=3\left(\frac{3}{7}\right)^i+9\left(\frac{3}{7}\right)^{i+1}+\delta=\frac{48}{7}\left(\frac{3}{7}\right)^i+\delta$. If player $j$ migrates to the other strategy $o_j$, thus joining her three children, her cost becomes $\widehat{c}_j(K_{-j}\diamond o_j)=16\left(\frac{3}{7}\right)^{i+1}+\delta=\frac{48}{7}\left(\frac{3}{7}\right)^i+\delta$. Thus, player $j$ has no incentive to deviate from $K$.

Consider a player $j$ using a resource $k_j$ belonging to level $h-1$. Player $j$ is sharing $k_j$ with her two siblings, thus $c_j(K)=3\left(\frac{3}{7}\right)^{h-1}$. The two children of player $j$ are sharing the resource $o_j$ belonging to level $h$, thus each of them is paying a cost of $2\frac 3 5\left(\frac{3}{7}\right)^{h-1}$. Finally, assume that the parent of player $j$ is paying a cost of $\delta$. It follows that $\widehat{c}_j(K)=3\left(\frac{3}{7}\right)^{h-1}+\frac{12}{5}\left(\frac{3}{7}\right)^{h-1}+\delta=\frac{27}{5}\left(\frac{3}{7}\right)^{h-1}+\delta$. If player $j$ migrates to the other strategy $o_j$, thus joining her two children, her cost becomes $\widehat{c}_j(K_{-j}\diamond o_j)=\frac{27}{5}\left(\frac{3}{7}\right)^{h-1}+\delta$. Thus, player $j$ has no incentive to deviate from $K$.

Consider a player $j$ using a resource $k_j$ belonging to level $i$, with $h\leq i\leq 2h-2$. Player $j$ is sharing $k_j$ with her sibling, thus $c_j(K)=2\frac 3 5\left(\frac{3}{7}\right)^{h-1}\left(\frac{2}{5}\right)^{i-h}$. The two children of player $j$ are sharing the resource $o_j$ belonging to level $i+1$, thus each of them is paying a cost of $2\frac 3 5\left(\frac{3}{7}\right)^{h-1}\left(\frac{2}{5}\right)^{i-h+1}$. Finally, assume that the parent of player $j$ is paying a cost of $\delta$. It follows that $\widehat{c}_j(K)=\frac 6 5\left(\frac{3}{7}\right)^{h-1}\left(\frac{2}{5}\right)^{i-h}+\frac{12}{5}\left(\frac{3}{7}\right)^{h-1}\left(\frac{2}{5}\right)^{i-h+1}+\delta
=\frac{54}{25}\left(\frac{3}{7}\right)^{h-1}\left(\frac{2}{5}\right)^{i-h}+\delta$. If player $j$ migrates to the other strategy $o_j$, thus joining her two children, her cost becomes $\widehat{c}_j(K_{-j}\diamond o_j)=9\frac 3 5\left(\frac{3}{7}\right)^{h-1}\left(\frac{2}{5}\right)^{i-h+1}+\delta=\frac{54}{25}\left(\frac{3}{7}\right)^{h-1}\left(\frac{2}{5}\right)^{i-h}+\delta$. Thus, player $j$ has no incentive to deviate from $K$.

Finally, consider a player $j$ using a resource $k_j$ belonging to level $2h-1$. Player $j$ is sharing $k_j$ with her sibling, thus $c_j(K)=2\frac 3 5\left(\frac{6}{35}\right)^{h-1}$. Player $j$ has no children, thus assuming that her parent is paying a cost of $\delta$, it follows that $\widehat{c}_j(K)=\frac 6 5\left(\frac{6}{35}\right)^{h-1}+\delta$. If player $j$ migrates to the other strategy $o_j$, used by no players in $K$, her cost becomes $\widehat{c}_j(K_{-j}\diamond o_j)=\frac 6 5\left(\frac{6}{35}\right)^{h-1}+\delta$. Thus, player $j$ has no incentive to deviate from $K$.

In order to bound the value of ${\sf SUM}(K)$, note that each resource from level $0$ to $h-1$ is used by three players, while each resource from level $h$ to $2h-1$ is used by two players. Thus, we obtain
\begin{eqnarray*}
{\sf SUM}(K) & = & 9\sum_{i=0}^{h-1}\left(3^i\left(\frac 3 7\right)^i\right)+4\sum_{i=h}^{2h-1}\left(3^h 2^{i-h}\frac 3 5\left(\frac 3 7\right)^{h-1}\left(\frac 2 5\right)^{i-h}\right)\\
& = & \frac{153}{2}\left(\frac 9 7\right)^{h-1}-\frac{36}{5}\left(\frac{36}{35}\right)^{h-1}-\frac{63}{2}.
\end{eqnarray*}

We upper bound the value of the social optimum with the social value of the profile $O$ in which all players choose the resource closest to the leaves. Note that, in this case, each resource from level $1$ to $2h$ is used by one player. Thus, we obtain
\begin{eqnarray*}
{\sf SUM}(O) & = & \sum_{i=1}^{h-1}\left(3^i\left(\frac 3 7\right)^i\right)+\sum_{i=h}^{2h-1}\left(3^h 2^{i-h}\frac 3 5\left(\frac 3 7\right)^{h-1}\left(\frac 2 5\right)^{i-h}\right)+6^h\frac 6 5\left(\frac{6}{35}\right)^{h-1}\\
& = & \frac{27}{2}\left(\frac 9 7\right)^{h-1}-\frac{9}{2}.
\end{eqnarray*}

Hence, for any $\epsilon>0$, there exists a sufficiently big $h$ such that $${\sf PoA}(T,\Gamma)\geq\frac{\frac{153}{2}\left(\frac 9 7\right)^{h-1}}{\frac{27}{2}\left(\frac 9 7\right)^{h-1}}-\epsilon=\frac{17}{3}-\epsilon.$$

\

\noindent{\bf Proof of Theorem \ref{lbpos}.}
For any fixed $\epsilon>0$, $\cal G$ is defined as follows. The $n$ players are partitioned into three subsets $P$, $P'$ and $P''$ such that $|P|=|P'|=n_1$ and $|P''|=n_2$ and there are $2(n_1^2+n_1+1)$ resources. Each player $i$, with $i\in P$, has two strategies, denoted with $k_i$ and $o_i$, each player $i$, with $i\in P'$, has two strategies, denoted with $k'_i$ and $o'_i$, and each player $i$, with $i\in P''$, has a unique strategy denoted with $s$. The resources are divided into six different types: there are $n_1$ resources of types $A$ and $B$, denoted with $A_i$ and $B_i$ for each $i\in [n_1]$, $n_1^2$ resources of types $C$ and $D$, denoted with $C_{ij}$ and $D_{ij}$ for each $i,j\in [n_1]$, and $1$ resource of types $E$ and $F$. Resource $A_i$ only belongs to $o_i$ for each $i\in P$, resource $B_i$ only belongs to $o'_i$ for each $i\in P'$, resource $C_{ij}$ belongs only to $k_i$ and $o'_j$ for each $i\in P$ and $j\in P'$, resource $D_{ij}$ belongs only to $k'_i$ and $o_j$ for each $i\in P'$ and $j\in P$, resource $E$ belongs only to $s$ and to $k_i$ for each $i\in P$ and resource $F$ belongs only to $s$ and to $k'_i$ for each $i\in P'$. Finally, each resource of type $A$ and $B$ has latency $\ell_A(x)=\ell_B(x)=(n_1+2n_2)x+\delta$, where $\delta>0$ is arbitrarily small, each resource of type $C$ and $D$ has latency $\ell_C(x)=\ell_D(x)=\frac x 2$ and the resources of type $E$ and $F$ have latency $\ell_E(x)=\ell_F(x)=2x$.

The matrix $\Gamma$ is such that $\gamma_{ii}=1$ for each $i\in [n]$, $\gamma_{ij}=1$ if and only if $i\in P$ and $j\in P'$ or $i\in P'$ and $j\in P$, while $\gamma_{ij}=0$ otherwise. Hence, $\Gamma$ is a boolean symmetric matrix which defines a restricted altruistic social context.

Note that the congestion of each resource of type $C$ and $D$ in any possible strategy profile is a value in $\{0,1,2\}$. In particular, for any strategy profile $S$, it holds
\begin{displaymath}
n_{C_{ij}}(S)=\left\{
\begin{array}{ll}
0 & \textrm{ if $i\in P$ chooses $o_i$ and $j\in P'$ chooses $k'_j$,}\\
2 & \textrm{ if $i\in P$ chooses $k_i$ and $j\in P'$ chooses $o'_j$,}\\
1 & \textrm{ otherwise.}\\
\end{array}\right.
\end{displaymath}
A similar characterization holds for $n_{D_{ij}}(S)$ by swapping the roles of the players in $P$ and $P'$.

We show that the strategy profile $K=\left((k_i)_{i\in P},(k'_i)_{i\in P'},(s)_{i\in P''}\right)$ is the unique Nash equilibrium of $({\cal G},\Gamma)$. Let $H$ be any strategy profile in which exactly $h\geq 1$ players in $P$ choose strategy $k$ (and, so, $n_1-h$ of them choose strategy $o$) and exactly $h'\geq 1$ players in $P'$ choose strategy $k'$ (and, so, $n_1-h'$ of them choose strategy $o'$). Since the players are symmetric, as well as the resources, all the players choosing the same type of strategy pay the same cost in $H$. Denote with $cost_k(H)$ the cost of any of the players in $P$ choosing strategy $k$ in $H$ and with $cost_o(H)$ the cost of any of the players in $P$ choosing strategy $o$ in $H$. Similarly, we denote with $cost'_k(H)$ the cost of any of the players in $P'$ choosing strategy $k'$ in $H$ and with $cost'_o(H)$ the cost of any of the players in $P'$ choosing strategy $o'$ in $H$.

Let us compute $cost_k(H)$. Without loss of generality, we can suppose that the first $h$ players in $P$ and the first $h'$ players in $P'$ are those choosing strategies of types $k$ and $k'$, respectively. Thus, we can focus on the cost paid by the first player belonging to $P$ in $H$. She is using resources $C_{1j}$ for each $1\leq j\leq n_1$ and resource $E$. The congestion of the latter is $n_2+h$. By exploiting the characterization of $n_{C_{ij}}(S)$ given above, we have that, of the $n_1$ resources of type $C$ used by the player, $n_1-h'$ of them have congestion $2$ (since there are $n_1-h'$ players in $P'$ using the strategy of type $o'$) and $h'$ of them have congestion $1$ (since there cannot be resources with congestion equal to $0$). Thus, it holds $cost_k(H)=\frac 1 2 (2n_1-h')+2(h+n_2)$.

Let us compute $cost_o(H)$. Again, we can focus on the cost paid by the last player belonging to $P$ in $H$. She is using resources $D_{i n_1}$ for each $i\in [n_1]$ and resource $A_{n_1}$. The congestion of the latter is $1$. By exploiting the characterization of $n_{D_{ij}}(S)$ given above, we have that, of the $n_1$ resources of type $D$ used by the player, $h'$ of them have congestion $2$ (since there are $h'$ in $P'$ using the strategy of type $k'$) and $n_1-h'$ of them have congestion $1$ (since there cannot be resources with congestion equal to $0$). Thus, it holds $cost_o(H)=n_1+2n_2+\delta+\frac 1 2 (n_1+h')$.

With a similar analysis, we obtain $cost'_k(H)=\frac 1 2 (2n_1-h)+2(h'+n_2)$ and $cost'_o(H)=n_1+2n_2+\delta+\frac 1 2 (n_1+h)$.

By the definition of $\Gamma$, each player in $P$ wants to minimize her cost plus the sum of the costs of all the players in $P'$. Thus, we get
$$\widehat{cost}_k(H)=\frac 1 2 (2n_1-h')+2(h+n_2)+h'\cdot cost'_k(H)+(n_1-h')\cdot cost'_o(H)$$ and $$\widehat{cost}_o(H)=n_1+2n_2+\delta+\frac 1 2 (n_1+h')+h'\cdot cost'_k(H)+(n_1-h')\cdot cost'_o(H).$$
Similarly, we obtain $\widehat{cost}'_k(H)$ and $\widehat{cost}'_o(H)$.

Let $H_1$ be the strategy profile obtained from $H$ when player $i\in P$ changes her strategy from $k_i$ to $o_i$, i.e., the profile in which the number of player in $P$ using the strategy of type $k$ is $h-1$. Note that, as long as $h\leq h'$, it holds $\widehat{cost}_k(H)<\widehat{cost}_o(H_1)$. Similarly, it is possible to establish that, as long as $h\geq h'$, it holds $\widehat{cost}'_k(H)<\widehat{cost}'_o(H_2)$, where $H_2$ is the strategy profile obtained from $H$ when player $i\in P'$ changes her strategy from $k'_i$ to $o'_i$, i.e., the profile in which the number of player in $P'$ using the strategy of type $k'$ is $h'-1$. Thus, in each strategy profile $H\neq K$, there always exists a player using a strategy of type $o$ or $o'$ who can improve by choosing the strategy of type $k$ or $k'$. This shows that $K$ is the only pure Nash equilibrium for $({\cal G},\Gamma)$.

Let us now compare ${\sf SUM}(K)$ with ${\sf SUM}(O)$, where $O=\left((o_i)_{i\in P},(o'_i)_{i\in P'},(s)_{i\in P''}\right)$. To this aim, note that each resource of type $A$ and $B$ has congestion $0$ in $K$ and $1$ in $O$, each resource of type $C$ and $D$ has congestion $1$ both in $K$ and $O$ and the resources of type $E$ and $F$ have congestion $n_1+n_2$ in $K$ and $n_2$ in $O$. Thus, we obtain
$${\sf PoS}({\cal G},\Gamma) = \frac{n_1^2+4(n_1+n_2)^2}{2n_1\left(n_1+2n_2+\delta\right)+n_1^2+4n_2^2}.$$

By choosing $n_1=2(1+\sqrt{2})n_2$ and $n_2$ sufficiently big, we get ${\sf PoS}({\cal G},\Gamma)\geq 1+1/\sqrt{2}-\epsilon$.

\subsection{Omitted Material from Section \ref{sec-special}}

\noindent{\bf Proof of Theorem \ref{ubposetero}.}
For $v\in [0,1/2]$, set $\theta=\frac{(\sqrt{3}+1)(1-v)}{\sqrt{3}-v(\sqrt{3}-1)}$, $x=\frac{3-2(1+\sqrt{3})v^2-(3-\sqrt{3})v}{2(2v^2-6v+3)}$ and $y_i=\frac{2(1+\sqrt{3})v^2-(1+3\sqrt{3})v+\sqrt{3}}{2v^2-6v+3}$ for each $i\in [n]$. With these values, the dual constraint becomes $(2v-1)f(K_e,O_e)\geq 0$, with
\begin{eqnarray*}
f(K_e,O_e) & := & K_e^2((\sqrt{3}-1)v+3-2\sqrt{3})-K_e(2O_e-\sqrt{3})((1+\sqrt{3})v-\sqrt{3})\\
& & \ \ +O_e(O_e-1)((5+3\sqrt{3})v-3-2\sqrt{3})
\end{eqnarray*}
which, for any $v\in [0,1/2]$, is non-negative when $f(K_e,O_e)\leq 0$.
Note that the discriminant of the equation $f(K_e,O_e)=0$, when solved for $K_e$, is $$v^2(18(2+\sqrt{3})-(32+16\sqrt{3})O_e)+v((48+16\sqrt{3})O_e-18(3+\sqrt{3}))-3(8O_e-9).$$ Such a quantity is always non-positive when $O_e\geq 2$, hence we are just left to check the cases of $O_e\in\{0,1\}$. For $O_e=0$, $f(K_e,O_e)\leq 0$ becomes $$K_e(v((\sqrt{3}-1)K_e+3+\sqrt{3})+(3-2\sqrt{3})K_e-3)\leq 0$$ which is always verified when $K_e\geq\frac{\sqrt{3}(\sqrt{3}-(1+\sqrt{3})v)}{(\sqrt{3}-1)v+3-2\sqrt{3}}$. Since the right-hand side of this inequality is never positive for any $v\in [0,1/2]$, we are done. For $O_e=1$, $f(K_e,O_e)\leq 0$ becomes $$K_e(v((\sqrt{3}-1)K_e+1-\sqrt{3})+(3-2\sqrt{3})K_e+2\sqrt{3}-3)\leq 0$$ which is always verified for any non-negative integer $K_e$.

For $v\in [1/2,1]$, set $\theta=\frac{3-\sqrt{3}-2v(2-\sqrt{3})}{2(1-v)}$, $x=\frac{1+2v-\sqrt{3}(2v-1)}{4(1-v)}$ and $y_i=\frac{(2v-1)(\sqrt{3}-1)}{2(1-v)}$ for each $i\in [n]$. (Note that, for $v=1$, the variables $\theta$, $x$ and $y_i$ are not correctly defined. In fact, in such a case, the price of stability is unbounded which implies that the dual program is unfeasible). With these values, the dual constraint becomes $\frac{1-2v}{v-1}f(K_e,O_e)\geq 0$, with
$$f(K_e,O_e):=K_e^2(1+\sqrt{3})-K_e(2O_e(\sqrt{3}-1)+1+\sqrt{3})+O_e(O_e(3\sqrt{3}-5)+3-\sqrt{3})$$ which, for any $v\in [1/2,1]$, is non-negative when $f(K_e,O_e)\geq 0$.
Note that the discriminant of the equation $f(K_e,O_e)=0$, when solved for $K_e$, is $$4O_e(1-\sqrt{3})+2+\sqrt{3}.$$ Such a quantity is always non-positive when $O_e\geq 2$, hence we are just left to check the cases of $O_e\in\{0,1\}$. For $O_e=0$, $f(K_e,O_e)\geq 0$ becomes $$K_e(K_e-1)\geq 0$$ which is always verified for any non-negative integer $K_e$. For $O_e=1$, $f(K_e,O_e)\geq 0$ becomes $$K_e(1+\sqrt{3})-K_e(3\sqrt{3}-1)+2\sqrt{3}-2\geq 0$$ which is always verified for any non-negative integer $K_e$.

\

\noindent{\bf Proof of Theorem \ref{lbposspecial}.}
For any fixed $\epsilon>0$ and $v\in [0,1/2]$, $\cal G$ is defined as follows. The $n$ players are partitioned into two subsets $P$ and $P'$ such that $|P|=n_1$ and $|P'|=n_2$ and there are $n_1^2+1$ resources. Each player $i$, with $i\in P$, has two strategies, denoted with $k_i$ and $o_i$, while each player in $P'$ has a unique strategy denoted with $s$. The resources are divided into three different types: there are $n_1$ resources of type $A$, denoted with $A_i$ for each $i\in [n_1]$, $n_1(n_1-1)$ resources of type $B$, denoted with $B_{ij}$ for each $i,j\in [n_1]$ with $i\neq j$, and $1$ resource of type $C$. Resource $A_i$ only belongs to $o_i$ for each $i\in P$, resource $B_{ij}$ belongs only to $k_i$ and $o_j$ for each $i,j\in P$ with $i\neq j$ and resource $C$ belongs only to $s$ and to $k_i$ for each $i\in P$. Finally, each resource of type $A$ has latency $\ell_A(x)=\frac{n_1+2n_2+1-2v}{2(1-v)}x+\delta$, where $\delta>0$ is arbitrarily small, each resource of type $B$ has latency $\ell_B(x)=\frac x 2$ and the resources of type $C$ has latency $\ell_C(x)=x$.

Note that the congestion of each resource of type $B$ in any possible strategy profile is a value in $\{0,1,2\}$. In particular, for any strategy profile $S$, it holds
\begin{displaymath}
n_{B_{ij}}(S)=\left\{
\begin{array}{ll}
0 & \textrm{ if $i$ chooses $o_i$ and $j$ chooses $k_j$,}\\
2 & \textrm{ if $i$ chooses $k_i$ and $j$ chooses $o_j$,}\\
1 & \textrm{ otherwise.}\\
\end{array}\right.
\end{displaymath}

We show that the strategy profile $K=\left((k_i)_{i\in P},(s)_{i\in P'}\right)$ is the unique Nash equilibrium of $({\cal G},\Gamma)$. Let $H$ be any strategy profile in which exactly $h\geq 1$ players in $P$ choose strategy $k$ (and, so, $n_1-h$ of them choose strategy $o$). Since the players are symmetric, as well as the resources, all the players choosing the same type of strategy pay the same cost in $H$. Denote with $cost_k(H)$ the cost of any of the players in $P$ choosing strategy $k$ in $H$ and with $cost_o(H)$ the cost of any of the players in $P$ choosing strategy $o$ in $H$.

Let us compute $cost_k(H)$. Without loss of generality, we can suppose that the first $h$ players in $P$ are those choosing strategies of type $k$. Thus, we can focus on the cost paid by the first player belonging to $P$ in $H$. She is using resources $B_{1j}$ for each $1\leq j\leq n_1$ with $j\neq 1$ and resource $C$. The congestion of the latter is $n_2+h$. By exploiting the characterization of $n_{B_{ij}}(S)$ given above, we have that, of the $n_1-1$ resources of type $B$ used by the player, $n_1-h$ of them have congestion $2$ (since there are $n_1-h$ players in $P$ using the strategy of type $o$) and $h-1$ of them have congestion $1$ (since there cannot be resources with congestion equal to $0$). Thus, it holds $cost_k(H)=\frac 1 2 (2n_1-h-1)+h+n_2$.

Let us compute $cost_o(H)$. Again, we can focus on the cost paid by the last player belonging to $P$ in $H$. She is using resources $B_{i n_1}$ for each $i\in [n_1-1]$ and resource $A_{n_1}$. The congestion of the latter is $1$. By exploiting the characterization of $n_{B_{ij}}(S)$ given above, we have that, of the $n_1-1$ resources of type $B$ used by the player, $h$ of them have congestion $2$ (since there are $h$ players in $P$ using the strategy of type $k$) and $n_1-h-1$ of them have congestion $1$ (since there cannot be resources with congestion equal to $0$). Thus, it holds $cost_o(H)=\frac{n_1+2n_2+1-2v}{2(1-v)}+\delta+\frac 1 2 (n_1+h-1)$.

By the definition of $\Gamma$, each player in $P$ wants to minimize $(1-v)$ times her cost plus the sum of the costs of all the players in the game multiplied by $v$. Thus, we get
\begin{displaymath}
\begin{array}{cl}
& \widehat{cost}_k(H)\\
= & (1-v)cost_k(H)+v\left((h-1)cost_k(H)+(n_1-h)cost_o(H)+n_2(n_2+h)\right)
\end{array}
\end{displaymath}
and
\begin{displaymath}
\begin{array}{cl}
& \widehat{cost}_o(H)\\
= & (1-v)cost_o(H)+v\left(h\cdot cost_k(H)+(n_1-h-1)cost_o(H)+n_2(n_2+h)\right).
\end{array}
\end{displaymath}

Let $H'$ be the strategy profile obtained from $H$ when a player $i\in P$ changes her strategy from $k_i$ to $o_i$, i.e., the profile in which the number of players in $P$ using the strategy of type $k$ is $h-1$. Note that it holds $\widehat{cost}_k(H)<\widehat{cost}_o(H')$. Thus, in each strategy profile $H\neq K$, there always exists a player using a strategy of type $o$ who can improve by choosing the strategy of type $k$. This shows that $K$ is the only pure Nash equilibrium for $({\cal G},\Gamma)$.

Let us now compare ${\sf SUM}(K)$ with ${\sf SUM}(O)$, where $O=\left((o_i)_{i\in P},(s)_{i\in P'}\right)$. To this aim, note that each resource of type $A$ has congestion $0$ in $K$ and $1$ in $O$, each resource of type $B$ has congestion $1$ both in $K$ and $O$ and the resource of type $C$ has congestion $n_1+n_2$ in $K$ and $n_2$ in $O$. Thus, we obtain
$${\sf PoS}({\cal G},\Gamma) = \frac{\frac 1 2 n_1(n_1-1)+(n_1+n_2)^2}{\left(\frac{n_1+2n_2+1-2v}{2(1-v)}+\delta\right)n+\frac 1 2 n_1(n_1-1)+n_2^2}.$$

By choosing $n_1=(1+\sqrt{3})n_2$ and $n_2$ sufficiently big, we get ${\sf PoS}({\cal G},\Gamma)\geq \frac{(\sqrt{3}+1)(1-v)}{\sqrt{3}-v(\sqrt{3}-1)}-\epsilon$.

\

For any fixed $\epsilon>0$ and $v\in [1/2,1]$, ${\cal G}'$ is defined as follows. The $n$ players are partitioned into two subsets $P$ and $P'$ such that $|P|=n_1$ and $|P'|=n_2$ and there are $n_1^2+1$ resources. Each player $i$, with $i\in P$, has two strategies, denoted with $k_i$ and $o_i$, while each player in $P'$ has a unique strategy denoted with $s$. The resources are divided into three different types: there are $n_1$ resources of type $A$, denoted with $A_i$ for each $i\in [n_1]$, $n_1(n_1-1)$ resources of type $B$, denoted with $B_{ij}$ for each $i,j\in [n_1]$ with $i\neq j$, and $1$ resource of type $C$. Resource $A_i$ only belongs to $k_i$ for each $i\in P$, resource $B_{ij}$ belongs only to $k_i$ and $o_j$ for each $i,j\in P$ with $i\neq j$ and resource $C$ belongs only to $s$ and to $o_i$ for each $i\in P$. Finally, each resource of type $A$ has latency $\ell_A(x)=\frac{n_1+2n_2+1-2v}{2(1-v)}x-\delta$, where $\delta>0$ is arbitrarily small, each resource of type $B$ has latency $\ell_B(x)=\frac x 2$ and the resources of type $C$ has latency $\ell_C(x)=x$.

Note that the congestion of each resource of type $B$ in any possible strategy profile is a value in $\{0,1,2\}$. In particular, for any strategy profile $S$, it holds
\begin{displaymath}
n_{B_{ij}}(S)=\left\{
\begin{array}{ll}
0 & \textrm{ if $i$ chooses $o_i$ and $j$ chooses $k_j$,}\\
2 & \textrm{ if $i$ chooses $k_i$ and $j$ chooses $o_j$,}\\
1 & \textrm{ otherwise.}\\
\end{array}\right.
\end{displaymath}

We show that the strategy profile $K=\left((k_i)_{i\in P},(s)_{i\in P'}\right)$ is the unique Nash equilibrium of $({\cal G},\Gamma)$. Let $H$ be any strategy profile in which exactly $h\geq 1$ players in $P$ choose strategy $k$ (and, so, $n_1-h$ of them choose strategy $o$). Since the players are symmetric, as well as the resources, all the players choosing the same type of strategy pay the same cost in $H$. Denote with $cost_k(H)$ the cost of any of the players in $P$ choosing strategy $k$ in $H$ and with $cost_o(H)$ the cost of any of the players in $P$ choosing strategy $o$ in $H$.

Let us compute $cost_k(H)$. Without loss of generality, we can suppose that the first $h$ players in $P$ are those choosing strategies of type $k$. Thus, we can focus on the cost paid by the first player belonging to $P$ in $H$. She is using resources $B_{1j}$ for each $1\leq j\leq n_1$ with $j\neq 1$ and resource $A_1$. The congestion of the latter is $1$. By exploiting the characterization of $n_{B_{ij}}(S)$ given above, we have that, of the $n_1-1$ resources of type $B$ used by the player, $n_1-h$ of them have congestion $2$ (since there are $n_1-h$ players in $P$ using the strategy of type $o$) and $h-1$ of them have congestion $1$ (since there cannot be resources with congestion equal to $0$). Thus, it holds $cost_k(H)=\frac{n_1+2n_2+1-2v}{2(1-v)}-\delta+\frac 1 2 (2n_1-h-1)$.

Let us compute $cost_o(H)$. Again, we can focus on the cost paid by the last player belonging to $P$ in $H$. She is using resources $B_{i n_1}$ for each $i\in [n_1-1]$ and resource $C$. The congestion of the latter is $n_1-h+n_2$. By exploiting the characterization of $n_{B_{ij}}(S)$ given above, we have that, of the $n_1-1$ resources of type $B$ used by the player, $h$ of them have congestion $2$ (since there are $h$ players in $P$ using the strategy of type $k$) and $n_1-h-1$ of them have congestion $1$ (since there cannot be resources with congestion equal to $0$). Thus, it holds $cost_o(H)=\frac 1 2 (n_1+h-1)+n_1-h+n_2$.

By the definition of $\Gamma$, each player in $P$ wants to minimize $(1-v)$ times her cost plus the sum of the costs of all the players in the game multiplied by $v$. Thus, we get
\begin{displaymath}
\begin{array}{cl}
& \widehat{cost}_k(H)\\
= & (1-v)cost_k(H)+v\left((h-1)cost_k(H)+(n_1-h)cost_o(H)+n_2(n_1-h+n_2)\right)
\end{array}
\end{displaymath}
and
\begin{displaymath}
\begin{array}{cl}
& \widehat{cost}_o(H)\\
= & (1-v)cost_o(H)+v\left(h\cdot cost_k(H)+(n_1-h-1)cost_o(H)+n_2(n_1-h+n_2)\right).
\end{array}
\end{displaymath}

Let $H'$ be the strategy profile obtained from $H$ when a player $i\in P$ changes her strategy from $k_i$ to $o_i$, i.e., the profile in which the number of players in $P$ using the strategy of type $k$ is $h-1$. Note that it holds $\widehat{cost}_k(H)<\widehat{cost}_o(H')$. Thus, in each strategy profile $H\neq K$, there always exists a player using a strategy of type $o$ who can improve by choosing the strategy of type $k$. This shows that $K$ is the only pure Nash equilibrium for $({\cal G},\Gamma)$.

Let us now compare ${\sf SUM}(K)$ with ${\sf SUM}(O)$, where $O=\left((o_i)_{i\in P},(s)_{i\in P'}\right)$. To this aim, note that each resource of type $A$ has congestion $1$ in $K$ and $0$ in $O$, each resource of type $B$ has congestion $1$ both in $K$ and $O$ and the resource of type $C$ has congestion $n_2$ in $K$ and $n_1+n_2$ in $O$. Thus, we obtain
$${\sf PoS}({\cal G}',\Gamma') = \frac{\left(\frac{n_1+2n_2+1-2v}{2(1-v)}-\delta\right)n+\frac 1 2 n_1(n_1-1)+n_2^2}{\frac 1 2 n_1(n_1-1)+(n_1+n_2)^2}.$$

By choosing $n_1=(1+\sqrt{3})n_2$ and $n_2$ sufficiently big, we get ${\sf PoS}({\cal G}',\Gamma')\geq \frac{3-\sqrt{3}-2v(2-\sqrt{3})}{2(1-v)}-\epsilon$.

\

\noindent{\bf Proof of Theorem \ref{ubpoaspecial}.}
For $\overline{v}\in [1/2,1]$, consider the dual solution such that $\theta=\frac{2-\underline{v}}{1-\overline{v}}$ and $x_i=\frac{1}{1-\overline{v}}$ for each $i\in [n]$. (Note that, for $\overline{v}=1$, $x_i$ and $\theta$ are not correctly defined. In fact, in such a case, the price of anarchy is unbounded which implies that the dual program is unfeasible). With these values, for each $e\in E$, the dual constraint becomes
$$\underline{v}O_e(O_e-1)+\overline{v}K_e(1-K_e)+O_e(K_e-2O_e+1)\leq 0.$$
For $O_e=0$, such an inequality becomes $\overline{v}K_e(K_e-1)\geq 0$ which is always verified for any non negative integer $K_e$ when $\overline{v}\geq 0$, while, for $O_e=1$, it becomes $\overline{v}K_e(K_e-2)-K_e+1\geq 0$ which is always verified for any non negative integer $K_e$ when $\overline{v}\geq\frac 1 2$.

For $O_e\geq 2$, the discriminant of the equation associated with the dual constrained, when solved for $K_e$, is $$4\underline{v}\overline{v}O_e(O_e-1)+\overline{v}^2+2\overline{v}O_e(3-4O_e)+O_e^2$$ which is non-positive when it holds
$$\overline{v}^2(4O_e^2-4O_e+1)+2\overline{v}O_e(3-4O_e)+O_e^2\leq 0.$$
Such an inequality is verified for any $$\overline{v}\in\left[\frac{O_e\left(4O_e-3-2\sqrt{3O_e^2-5O_e+2}\right)}{4O_e^2-4O_e+1},
\frac{O_e\left(4O_e-3+2\sqrt{3O_e^2-5O_e+2}\right)}{4O_e^2-4O_e+1}\right].$$
Since, for any $O_e\geq 2$, such an interval is contained in the interval $\left[\frac 2 9,1+\frac{\sqrt{3}}{2}\right]$, the proposed dual solution is feasible.

For $\overline{v}\in [0,1/2]$, consider the dual solution such that $\theta=\frac{5+2\overline{v}-3\underline{v}}{2-\overline{v}}$ and $x_i=\frac{3}{2-\overline{v}}$ for each $i\in [n]$. With these values, for each $e\in E$, the dual constraint becomes
$$3\underline{v}O_e(O_e-1)-\overline{v}(K_e^2-3K_e+2O_e^2)-K_e^2+3K_eO_e-O_e(5O_e-3)\leq 0.$$
For $O_e=0$, such an inequality becomes $K_e^2+\overline{v}K_e(K_e-3)\geq 0$ which is always verified for any non negative integer $K_e$ when $\overline{v}\leq \frac 1 2$, while, for $O_e=1$, it becomes $(1+\overline{v})(K_e^2-3K_e+2)\geq 0$ which is always verified for any non negative integer $K_e$ when $\overline{v}\geq 0$.

For $O_e\geq 2$, the discriminant of the equation associated with the dual constrained, when solved for $K_e$, is $$12\underline{v}(1+\overline{v})O_e(O_e-1)+\overline{v}^2(9-8O_e^2)+2\overline{v}O_e(15-14O_e)-O_e(11O_e-12)$$ which is non-positive when it holds
$$\overline{v}^2(4O_e^2-12O_e+9)+2\overline{v}O_e(9-8O_e)-O_e(11O_e-12)\leq 0.$$
For $O_e\geq 2$, the quantity $-O_e(11O_e-12)$ is always negative, hence, in order to show the validity of the above inequality, we only need to prove that $$\overline{v}^2(4O_e^2-12O_e+9)+2\overline{v}O_e(9-8O_e)\leq 0.$$ Such an inequality is always verified when $\overline{v}\leq\frac{16O_e^2-18O_e}{4O_e^2-12O_e+9}$. Since, for $O_e\geq 2$, the right-hand side of this inequality is lower bounded by $4$, the proposed dual solution is feasible.

\end{document}